\documentclass[%
 reprint,
 amsmath,amssymb,
 aps,
]{revtex4-2}

\usepackage{graphicx}
\usepackage{dcolumn}
\usepackage{bm}
\usepackage{amsthm}

\usepackage{graphicx}
\usepackage[colorlinks=true, allcolors=blue]{hyperref}
\usepackage{graphicx} 
\usepackage{braket}
\usepackage{amsfonts}
\usepackage{amsmath}
\usepackage{dsfont}
\usepackage{amsmath,amssymb}
\usepackage{tikz-cd}
\usepackage{quiver}
\usepackage{physics}
\usepackage{qcircuit}
\usepackage{dsfont}
\newtheorem{theorem}{Theorem}
\usepackage{xcolor}
\usepackage{soul}

\usepackage{algorithm}
\usepackage{algpseudocode}
\usepackage{mathtools}
\usepackage{lipsum}
\usepackage{float}
\usepackage[utf8]{inputenc}
\usepackage[justification=raggedright,singlelinecheck=false]{caption}

\usepackage{etoolbox}    

\begin{document}

\preprint{APS/123-QED}

\title{Dynamics and Computation in Linear Open Quantum Walks}

\author{Pedro Linck Maciel}
 \email{Contact author: pedro.linck@ufpe.br}
 \affiliation{ Departamento de Física, Centro de Ciências Exatas e da Natureza, Universidade Federal de Pernambuco, Recife-PE, Brazil}
\author{Nadja K. Bernardes}%
 \email{Contact author: nadja.bernardes@ufpe.br}
 \affiliation{ Departamento de Física, Centro de Ciências Exatas e da Natureza, Universidade Federal de Pernambuco, Recife-PE, Brazil}



\date{\today}

\begin{abstract}
Open Quantum Walks (OQW) are a type of quantum walk governed by the system's interaction with its environment. We explore the time evolution and the limit behavior of the OQW framework for Quantum Computation and show how we can represent random unitary quantum channels, such as the dephasing and depolarizing channels, in this model. We also develop a simulation protocol with circuit representation for this model, which is heavily inspired by the fact that graphs represent OQW and are, thereby, local (meaning that the state in a particular node interacts only with its neighborhood). We obtain asymptotic advantages in the system's dimension, circuit depth, and CNOT count compared to other simulation methods.

\end{abstract}

\maketitle


\section{\label{sec:level1}Introduction}
Quantum computation has emerged as a promising research area with several practical applications, capable of processing information much faster than classical computers and with fewer physical resources \cite{QCI,shor97,grover96}. However, quantum computation has its limitations. The experimental development of practical and scalable quantum computers that are robust to noise and decoherence is still a challenge. No physical system is fully isolated, but many can be approximated as isolated systems. This is not the case with most quantum systems: they are very fragile and, therefore, very susceptible to environmental interactions, leading to loss of information and reducing the possibility of performing complex computational tasks.
\par The traditional approach to solving these problems is to use Quantum Error Correction, which is very demanding. A new approach to this problem has been proposed using the environment and its dissipation processes as the medium for performing the computation, known as Dissipative Quantum Computation. Several different approaches have been proposed to perform quantum computation using dissipative systems \cite{Cirac09, Petr12, Petr13, Qiang24, korkmaz20}.

\par An efficient way to implement dissipative quantum computing algorithms is to use the (discrete) Open Quantum Walk (OQW) framework \cite{Petr12}. This framework is modeled as a discrete and stochastic quantum walk \cite{Qiang24} on a graph in which the environment entirely operates the jump between nodes, changing the walker's location on the graph and its internal state. Different types of discrete and continuous quantum walks have been proposed for several applications such as cryptography \cite{abd19}, search algorithms \cite{QWSA}, simulation techniques \cite{lloyd08}, graph-theoretical questions \cite{Berry11, douglas08}, algebraic questions (checking matrix multiplication \cite{Buhrman06}, testing group commutativity \cite{magniez07},
subset sum \cite{becker11} and triangle finding \cite{magniez2007quantum}), and machine learning \cite{dernbach19}. Experimental results in the field of quantum walks and quantum walk based algorithms are also atracting much attention \cite{tang22,wang22,wang24,qu24,zheng24,Lorz19,bocanegra25,sahu24}. Many quantum walk models are coin-based models defined by a coin space and a coin operator that controls the evolution of the walker as the quantum analog of flipping a coin and deciding the path to take in a classical random walk \cite{QWSA}. In \cite{Petr12(1)}, it is shown that unitary quantum walks can be implemented as a modification of the OQW realization procedure. Therefore, Open Quantum Walks can be seen as a natural (and more realistic) extension of the most used models for applications, opening bright paths to be explored. As an example, a quantum algorithm for vertex ranking in graphs using discrete OQW has recently been proposed \cite{Dutta25}.

\par Quantum simulations are also an important area of quantum computation since many questions about complex quantum systems, such as the Fenna-Matthews-Olson protein complex \cite{Kais22} and the Avian Compass problem \cite{Kais23}, do not have an explicit analytical solution \cite{Kais24, Suri23}. Therefore, we also need efficient ways to simulate different types of quantum systems based on some general assumptions about the problem.

\par The main contribution of this work is to describe the evolution in time of this model of quantum computation, answer some theoretical questions (such as the steady state of the walk and an approximation for the running time of simulations and experiments) and propose an efficient algorithm concerning the system dimension, CNOT count, and depth (which we also generalize to a larger class of OQW models) to simulate this quantum computation model on perfect quantum computers. 

\par In Section \ref{sec:theoretical}, we study the time evolution of the OQW model under some restrictions. In Section \ref{sec:realization}, we show how to implement an important class of quantum channels (including dephasing and depolarizing channels) in this model. In Section \ref{sec:simulation}, we develop the simulation protocol based on graph locality and compare it with well-known simulation methods in the literature. We give our final considerations to the work in Section \ref{sec:conclusion}. The mathematical statements and proofs used in the text are in Appendix \ref{app:proofs}, and the description of the simulation protocol algorithm is in Appendix \ref{app:general}.

\newpage

\section{Theoretical considerations}
\label{sec:theoretical}
In this section, we explore some theoretical aspects of the evolution of open quantum walks. We derive the steady-state solution, approximate the evolution for a specific time interval, and estimate the number of steps required in the OQW model of quantum computation to design simulations and experiments.
\subsection{Open Quantum Walks}
A quantum walk is a type of quantum evolution that takes place on a graph. The vertices of the graph form the basis of possible places for the walker, and the edges give us the information about how the walker can jump from one node to another. When we talk about the dynamics of a quantum walk, we are talking about how the state of the system changes when we take steps. There is a zoo of quantum walks with the most distinguished properties \cite{QWSA}. 
In our case, we are interested in open quantum walks \cite{Petr12(1)}. The main idea behind open quantum walks is that we can jump from one node of the graph to another while changing the internal state of the walker by an arbitrary open quantum evolution, depending on the edge in which the jump is made. One can see a visual representation of this process in Fig. \ref{fig:OQW}. 

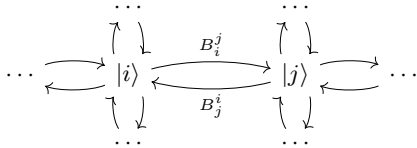
\begin{figure}[b]
\[\begin{tikzcd}[row sep=scriptsize]
	& \cdots && \cdots \\
	\cdots & {\ket{i}} && {\ket{j}} & \cdots \\
	& \cdots && \cdots
	\arrow[curve={height=-6pt}, from=1-2, to=2-2]
	\arrow[curve={height=-6pt}, from=1-4, to=2-4]
	\arrow[curve={height=-6pt}, from=2-1, to=2-2]
	\arrow[curve={height=-6pt}, from=2-2, to=1-2]
	\arrow[curve={height=-6pt}, from=2-2, to=2-1]
	\arrow["{B_i^j}", curve={height=-6pt}, from=2-2, to=2-4]
	\arrow[curve={height=-6pt}, from=2-2, to=3-2]
	\arrow[curve={height=-6pt}, from=2-4, to=1-4]
	\arrow["{B_j^i}", curve={height=-6pt}, from=2-4, to=2-2]
	\arrow[curve={height=-6pt}, from=2-4, to=2-5]
	\arrow[curve={height=-6pt}, from=2-4, to=3-4]
	\arrow[curve={height=-6pt}, from=2-5, to=2-4]
	\arrow[curve={height=-6pt}, from=3-2, to=2-2]
	\arrow[curve={height=-6pt}, from=3-4, to=2-4]
\end{tikzcd}\]\caption{\label{fig:OQW} An arbitrary open quantum walk can be represented by this visual diagram. If there is a omitted edge in a particular diagram, this means that the corresponding operator $B_i^j$ is zero.}
\end{figure}

\par The dynamics of an open quantum system with dimension $d$ can be described by a set of at most $d^2$ (not necessarily unitary) operators $\{K_i\}$, called Kraus operators \cite{QCI}. They transform the state $\rho$ into the state $\sum K_i \rho K_i^\dagger$ under the condition $\sum K_i^\dagger K_i = 1$. In the study of quantum computation, we deal with closed systems and unitary evolutions, so one question that naturally arises is: How can we simulate the evolution of open quantum systems using unitary operations?

    \par To formally define an open quantum walk, we need to define $\mathcal{H}$, the (internal) Hilbert space of the walker, and $\mathcal{G}$, the Hilbert space of the graph, where the walk is made. We mark a basis $\{\ket{i}\}_{i \in G}$ for the graph space, where $G$ is the set of vertices of the underlying graph. For each node $\ket{i}$ of the graph, we have a set of linear operators (not necessarily unitary nor positive) $\{B_i^j\}_{j \in G}$, $B_i^j\colon \mathcal{H} \to \mathcal{H}$ satisfying 
\begin{equation}
    \sum_j B_i^{j \dagger} B_i^j = \mathds{1}.
\end{equation}
\par The operator $B_i^j$ is the operator that corresponds to the internal evolution of the walker when jumping from node $\ket{i}$ to node $\ket{j}$. We need to define an operator that also carries the information of which edge is being jumped:
\begin{equation}
    M_i^j = B_i^j \otimes \ket{j}\bra{i}
\end{equation}
where $M_i^j$ are Kraus operators $M_i^j\colon \mathcal{H} \otimes \mathcal{G} \to \mathcal{H} \otimes \mathcal{G}$. It can be easily seen that
\begin{equation}
\sum_{i,j} M_i^{j \dagger} M_i^j = \mathds{1}. 
\end{equation}
\par Here, we denote $\mathcal{B}(\mathcal{V})$ as the convex space of the density matrices in the Hilbert space $\mathcal{V}$. Given a density matrix $\rho \in \mathcal{B}(\mathcal{H} \otimes \mathcal{G})$ to be the initial state of the walk, the following recursion relation governs the evolution of this system:
\begin{equation}
\label{eqn:recursiveoqw}
    \begin{cases}
      \rho^{[0]} = \rho \\
      \mathcal{M}(\rho^{[n]}) = \rho^{[n+1]} = \sum_{i,j} M_i^j \rho^{[n]} M_i^{j \dagger}
    \end{cases}       
\end{equation}
where $\rho^{[n]}$ is the state of the walk after $n$ steps. It was shown in Ref.\cite{Petr12(1)} that if $\rho = \sum_{i,j} \rho_{ij}\otimes \ket{i}\bra{j}$ for some positive trace-class matrices $\rho_{ij}$ satisfies $\sum_i\text{Tr}(\rho_{ii}) = 1$ then, after any $n \geq 1$ iterations, $\rho^{[n]} = \sum_i \rho_i \otimes \ket{i} \bra{i}$ for some $\rho_i$ determined by the evolution of the initial conditions and satisfying $\sum_i \text{Tr}(\rho_{i}) = 1$. Therefore, for studying the evolution of the system, we can assume, without loss of generality, that the initial condition has a diagonal form 
\begin{equation}
    \rho = \sum_i \rho_i \otimes \ket{i}\bra{i}.
\end{equation}
It is also important to note that the open quantum walk diagram in Fig.\ref{fig:OQW}, together with an initial state, is sufficient to determine the whole evolution. In a given graph diagram, if the edge drawn from $\ket{i}$ to $\ket{j}$ is marked with a given $C_i^j$, this can be translated to the OQW evolution given in Eq. \ref{eqn:recursiveoqw} as $M_i^j = C_i^j \otimes \ket{j}\bra{i}$. Theferore, given the graphic representation of all edges and operators between them (all omitted edges have corresponding Kraus operators equal to the zero operator), one can build all $M_i^j$. As a concrete example, in Fig. \ref{fig:linearOQW} we have $M_0^0 = \sqrt{\lambda} I \otimes \ket{0}\bra{0}$, $M_0^1 = \sqrt{\omega}U_0 \otimes \ket{1}\bra{0}$, $M_1^0 = \sqrt{\lambda} U_0^\dagger \otimes \ket{0}\bra{1}$, $M_1^1 = 0$ (if $N > 2$ there is no directed edge from $\ket{1}$ to $\ket{1}$), etc. for a given finite collection of unitaries $\{U_i\}$.

\subsection{Quantum Computation using OQW}
Here in this work, we will use a linear graph topology of OQW, where all $B_i^j$ operators are multiples of unitaries. In this work, the term topology refers to the graph topology, i.e., the connectivity of the graph. What we call computation here is the preparation of quantum states \cite{QWSA}, i.e., a sequence of quantum operations (not necessarily unitary) done in a determined quantum state to reach another quantum state. The definition of this model via its diagram representation is shown in Fig. \ref{fig:linearOQW}.
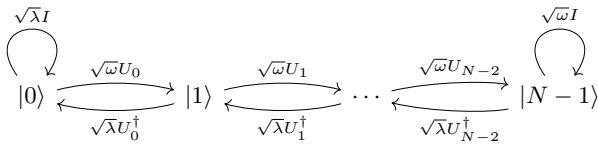
\begin{figure}[b]
\[\begin{tikzcd}
	{\ket{0}} && {\ket{1}} && \cdots && {\ket{N-1}}
	\arrow["{\sqrt{\lambda} I}", from=1-1, to=1-1, loop, in=55, out=125, distance=10mm]
	\arrow["{\sqrt{\omega} U_0}", curve={height=-6pt}, from=1-1, to=1-3]
	\arrow["{\sqrt{\lambda} U_0^\dagger}", curve={height=-6pt}, from=1-3, to=1-1]
	\arrow["{\sqrt{\omega} U_1}", curve={height=-6pt}, from=1-3, to=1-5]
	\arrow["{\sqrt{\lambda} U_1^\dagger}", curve={height=-6pt}, from=1-5, to=1-3]
	\arrow["{\sqrt{\omega} U_{N-2}}", curve={height=-6pt}, from=1-5, to=1-7]
	\arrow["{\sqrt{\lambda} U_{N-2}^\dagger}", curve={height=-6pt}, from=1-7, to=1-5]
	\arrow["{\sqrt{\omega} I}", from=1-7, to=1-7, loop, in=55, out=125, distance=10mm]
\end{tikzcd}\]
\caption{\label{fig:linearOQW} The diagram corresponding to the model of quantum computation based on OQW. Each $U_i$ is a unitary operator, and $\omega, \lambda \geq 0$ are such that $\omega + \lambda = 1$.}
\end{figure}
This diagram means that at node $\ket{i}$, we have the probability $\omega$ of jumping right and changing the internal state by applying $U_i$, and the probability $\lambda$ of jumping left and changing the internal state applying by $U_{i-1}^\dagger$ (except for the boundaries, where the state remains unchanged). The result of the computation will be detected (with some probability) in the last node in the long-term run. Each $U_i$ can be seen as a layer of a quantum circuit.
\par If we start at state $\ket{\psi}\bra{\psi} \otimes \ket{0} \bra{0}$, the state on the $n$-th iteration will be (using the convention $U_{-1} = \mathds{1})$
\begin{equation}
\label{eqn:niteration}
\rho^{[n]} = \sum_i p_{i}^{(n)} U_{i-1}\cdots U_0\ket{\psi}\bra{\psi}U_0^\dagger \cdots U_{i-1}^\dagger \otimes \ket{i}\bra{i}
\end{equation}
where $p_i^{(n)}$ is the probability of finding the walker in node $i$ after $n$ steps. Note that the only thing that changes from one iteration to another is the set of weights $p_i^{(n)}$. This means that no quantum walk effects can be seen observing only the probability distribution across the graph. Therefore, the evolution of the model can be described with an underlying classical random walk in this particular choice of OQW and initial state, in such a way that we recover a classical Markov chain \cite{Petr12(1)}. With that being said, we can then calculate some properties of this model based on classical random-walk techniques. We define $T$ as the transition matrix for this topology of OQW, i.e., $T_{ij}$ is the probability of jumping from node $j$ to node $i$:
\begin{equation}
T=\left[\begin{array}{cccccc}
\lambda & \lambda & 0 & & & \\
\omega & 0 & \lambda & &  & \vdots \\
0 & \omega & 0 & & & \vdots \\
\vdots & 0 & \omega & & & \vdots \\
\vdots & & 0 & \ddots & & \vdots \\
 &  &  & & \ddots &  \\
 &  & 
\end{array}\right]
\end{equation}
and the steady state $\vec{x}$ can be calculated by: 
\begin{equation}
\label{eqn: steadystate}
T \vec{x}=\vec{x} \implies x_m = \dfrac{a^m(a-1)}{a^N-1}
\end{equation}
where $a = \dfrac{\omega}{\lambda} = \dfrac{\omega}{1-\omega}$. Therefore, steady-state probabilities, viewed as a function $p = p(m) \propto a^m$, where $m$ is a node $1 \leq m \leq N-1$, have an exponential shape. A simulation is shown in Fig. \ref{fig:histogram}, where we compared the result of the iteration with $\omega = 2/3$ after $1000$ steps to the steady state calculated above, ensuring that they are practically equal at this step.

\begin{figure*}[tb]
\scalebox{0.7}{\input{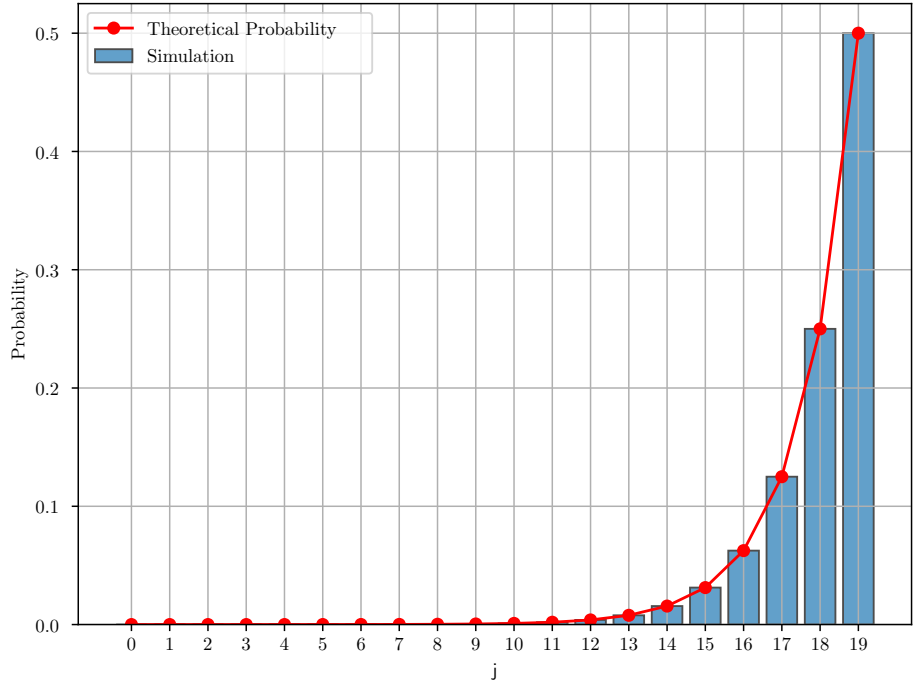}}
\caption{\label{fig:histogram} Comparison of the steady state solution of a linear OQW with its simulated values. The parameter are $N = 20$, $\omega = 2/3$ and $1000$ steps. The horizontal axis represents the nodes on the graph from 0 to 19; the vertical axis represents the probability of being at that node; the red dots represent the probability found theoretically in equation \ref{eqn: steadystate}, and the blue bins represent the frequency on a simulation of the walk.}
\end{figure*}
In this model, the computation is successful, by definition, if the result is detected at the last node, so we usually assume $\omega \geq \lambda$. In the long term, using Eq. (\ref{eqn: steadystate}), this occurs with probability
\begin{equation}
p_{suc} = \dfrac{a^{N-1}(a-1)}{a^N-1} = \dfrac{a^N - a^{N-1}}{a^N-1}.
\end{equation}
\par If we want $p_{suc} \geq \eta$ for a given $\eta > 0$, it is sufficient (see Appendix \ref{app:proofs}) to ask for 
\begin{equation}
    \omega \geq \dfrac{1}{2-\eta}.
\end{equation}
\par Note that this formula does not depend on the number of nodes $N = |G|$ ($G$ denotes the underlying graph). As an example, if we want at least $1/2$ probability of success in the long term, it is sufficient to ask for $\omega \geq 2/3$ (which occurs in Fig. \ref{fig:histogram}).
\par Suppose now that $N=2$ and that the initial state is
\begin{equation}
    \rho^{[0]} = p \rho_0 \otimes \ket{0}\bra{0} + (1-p) \rho_1 \otimes\ket{1}\bra{1}
\end{equation}
where $\rho_0,\rho_1 \in \mathcal{B}(\mathcal{H})$. After one single step, we have
\begin{equation}
\begin{split}
\rho^{[1]} & = [\lambda p \rho_0 + (1-p)\lambda U^{\dagger} \rho_1 U]\otimes \ket{0}\bra{0} \\
& + [\omega p U\rho_0 U^\dagger + \omega (1-p) \rho_1] \otimes \ket{1}\bra{1}.
\end{split}
\end{equation}
It can be shown (see Appendix \ref{app:proofs}) that, in this case, we always have
\begin{equation}
    \rho^{[n]} = \rho^{[1]}, \forall n \geq 1
\end{equation}
meaning that after just one step, the steady state is reached.
\subsection{Time evolution of linear OQW}

Now we proceed to approximate the solution for $p_m^{(n)}$ in Eq. \ref{eqn:niteration}, where $0<m<N-1$. Suppose $P(m,n)$ is the probability of reaching the $m$th node in $n$ steps. Then we have, by the definition of the OQW:
\begin{equation}
\label{mastereq}
P(m,n+1) = \omega P(m-1,n) + \lambda P(m+1,n)
\end{equation}
\par We are going to expand this equation in Taylor series, supposing that $N>>1$ in such a way that we can take the continuum limit \cite{SP,HSM}. Substituting
\begin{equation}
\begin{cases}
P(m \pm \Delta m,n) = P(m,n) \pm \partial P/\partial m + \partial^2 P/\partial m^2 \\
 P(m, n \pm \Delta n) = P(m,n) \pm \partial P/\partial n
\end{cases}
\end{equation}
in Eq. \ref{mastereq}, we obtain the differential equation 
\begin{equation}
    \dfrac{\partial P}{\partial n} = -v\dfrac{\partial P}{\partial m} + D\dfrac{\partial^2 P}{\partial m^2} 
\end{equation}
which has the following solution:
\begin{equation}
    P(m,n) = \dfrac{1}{\sqrt{4\pi D n}}e^{-(m-vn)^2/4Dn},
\end{equation}
where $D = 1/2$ is the diffusion coefficient and $v = \omega - \lambda = 2\omega - 1$ is the drift velocity. This solution is a normal distribution that moves in a line with velocity $v$ and has a dispersion rate $D$. This is true for a certain period, after which the distribution begins to deform into the exponential shape of the steady-state solution. \\
\par The time that the distribution takes to reach the right boundary is 
\begin{equation}
    n_{steps} = \dfrac{N}{v} = \dfrac{N}{2\omega - 1}.
\end{equation}
\par We can then estimate that after the $n_{steps}$ steps, we can finish our computation since this is the step in which the mean value of the distribution reaches the node that is important for the computation to be successful. A graph of the probability distribution on the nodes for several different times of the evolution can be seen in Fig. \ref{fig:probdist}, with $\omega = 2/3, \; N = 100$, where we can see that our approximate solution is valid for a certain time interval. Our estimated computation time for this case is $n_{steps} = 300$, which gives us $p_{99}^{(300)} \approx 0.3$, which is a good probability of success in the design of experiments and simulations.

\par It is also interesting to note that 
\begin{equation}
    \omega \geq \dfrac{1}{2 - \eta} \implies n_{steps} \leq \dfrac{N(2-\eta)}{\eta}.
\end{equation} 
We also note that if we start with a state $\sum_i q_i \ket{\psi_i}\bra{\psi_i} \otimes \ket{i}\bra{i}$, this estimate still holds, since the state already starts dispersed through the nodes, and it is easier to reach the final node in this way than starting condensed at the initial node.

\begin{figure*}[tb]
\scalebox{0.7}{\input{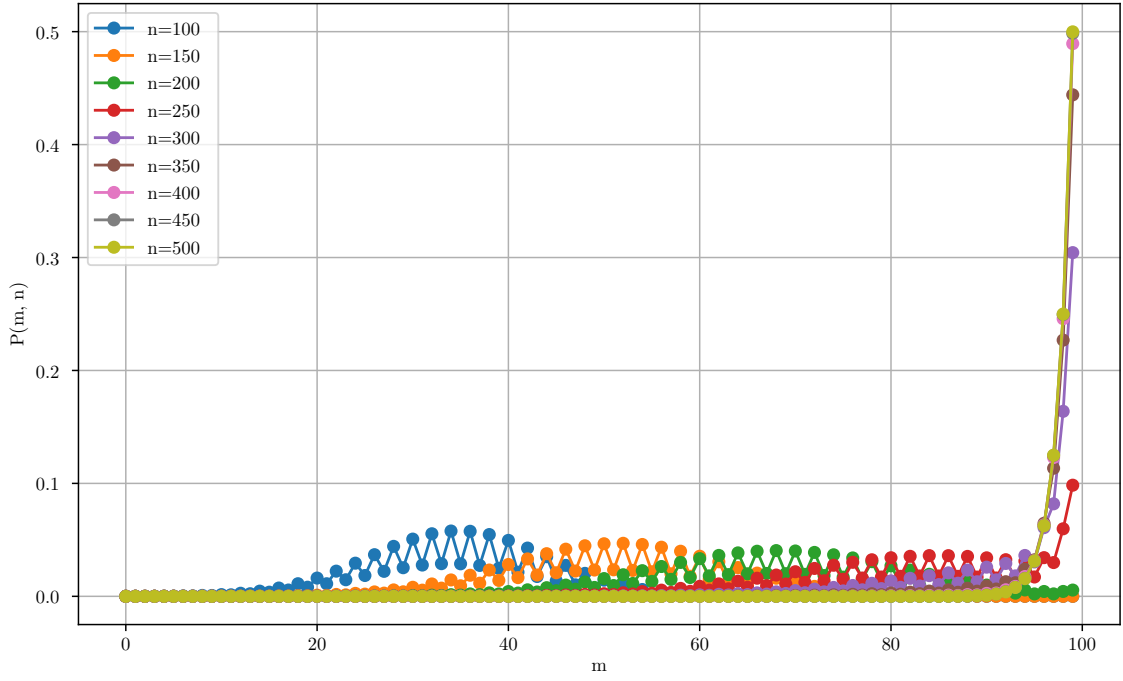}}
\caption{\label{fig:probdist} Simulation of the evolution of an OQW with $N = 100$, $\omega = 2/3$. The horizontal axis represents the nodes of the graph, from 0 to 99. The vertical axis represents the probability that the walker is on the given node after $n$ steps of the walk. We simulate for $n$ between 100 and 500, with intervals of 50 steps between each curve.}
\end{figure*}

\section{Realization of Quantum Channels}
\label{sec:realization}
Let us show how to implement any convex combination of unitaries in the OQW model of quantum computation and derive efficient ways to implement the dephasing and depolarizing quantum channels on OQW.
\subsection{Convex Combination of Unitaries}
\label{subsec:convexcomb}
We start with a Theorem concerning the recursive relations of the state after performing $n$ steps in the OQW: 
\begin{theorem}
\label{thm:recursive}
If the initial state in the OQW is of the form 
\begin{equation}
    \sum_i \rho_i\otimes \ket{i}\bra{i},
\end{equation}
then after $n$ steps, we have a state of the form
\begin{equation}
    \sum_i [a_0^{(i,n)}\rho_0^{(i)} +\cdots + a_{N-1}^{(i,n)}\rho_{N-1}^{(i)}]\otimes \ket{i}\bra{i}
\end{equation}
where
\begin{equation}
\begin{cases}
a_j^{(i,n+1)} = \omega a_j^{(i-1,n)} + \lambda a_j^{(i+1,n)} & \text{if $0 < i < N-1$}  \\
a_j^{(0,n+1)} = \lambda a_j^{(0,n)} + \lambda a_j^{(1,n)} & \text{if $i = 0$} \\
a_j^{(N-1,n+1)} = \omega a_j^{(N-2,n)} + \omega a_j^{(N-1,n)} & \text{if $i = N-1$} 
\end{cases}
\end{equation}
and
\begin{equation}
\rho_j^{(i)} = \begin{cases} \rho_j  & \text{if $i = j$} \\
U_{[j,i]}\rho_j U_{[j,i]}^\dagger & \text{if $i > j$} \\
U_{[i,j]}^\dagger \rho_j U_{[i,j]} & \text{if $j > i$} 
\end{cases}
\end{equation}
where we defined $U_{[i,j]} = U_j U_{j-1}\dots U_{i+1}U_i$.
\end{theorem}
This theorem is a direct calculation, entirely analogous to the calculation made in Theorem \ref{thm:correctness} in the Appendix \ref{app:proofs}. This Theorem is better understood as follows: Instead of evolving the whole sum $\sum_i \rho_i \otimes \ket{i}\bra{i}$, we first evolve each $\rho_j \otimes \ket{j}\bra{j}$ and then sum the results of their evolution. This can obviously be done since the evolution is linear, but it is important to emphasize the recursive relations for the contribution of the state of each node.

\par From the recursive relations in Theorem \ref{thm:recursive} we have $a_j^{(i,n)}\to p_j$ as $n \to \infty$ for all $i$, where $\{p_j\}$ are the steady-state probabilities. It is important to note that the convergence of these coefficients is independent of the initial state. We define $\ket{\psi}^{[n]} = U_{n}\cdots U_0\ket{\psi} = U^{[n]}\ket{\psi}$ for $n \geq 0$ and $U^{[-1]} = \mathds{1}$ for simplicity of notation. Therefore, in the long-term run, the part of the state corresponding to the last node is
\begin{equation}
    p_{N-1}(U^{[N-2]}\rho_{0}U^{[N-2]\dagger}+\cdots + \rho_{N-1}) \otimes \ket{N-1}\bra{N-1}.
\end{equation}
We conclude that after we post-select the result in this node, we obtain the state 
\begin{equation}
U^{[N-2]}\rho_{0}U^{[N-2]\dagger}+\cdots + \rho_{N-1}.
\end{equation}
Choosing $\rho_j, U_j$ appropriately (there is more than one way), this state can be made any convex combination
\begin{equation}
    \sum_i q_i V_i \rho V_i^\dagger
\end{equation}
of chosen unitaries $V_i$, where $\rho$ is a chosen density matrix and $q_i \in [0,1]$, $\sum_iq_i = 1$. In this way, we can express many quantum channels in the OQW framework of quantum computation in many different forms, some more efficient than others. This type of convex combination is called a Random Unitary Evolution. Observe that this procedure does not give us the most efficient way to implement convex combinations of unitaries in OQW; it only shows us that we can implement them in principle. In this way, we can also implement any single qubit unital map (i.e., quantum maps possessing the identity operator as a fixed point), as every single qubit unital map is a Random Unitary Evolution \cite{Chru13,Chru15}. Now we proceed to implement some important quantum channels in the OQW framework.

\subsection{Dephasing and Depolarizing channels}
\label{channels}
We can model single-qubit dephasing and depolarizing channels \cite{QCI} using this quantum computation model. We use the Pauli representation of the channels because Pauli operations are typically physically implementable on quantum computers. The dephasing channel is governed by the quantum operation
\begin{equation}
    \Delta_{\text{deph}}(\rho) = (1-p) \rho + p X\rho X,
\end{equation}
which admits a representation with Kraus operators $K_0 = \sqrt{1-p}\mathds{1}$ and $K_1 = \sqrt{p} X$.

\par The depolarizing channel is governed by 
\begin{equation}
    \Delta_{\text{depol}}^{(\lambda)}(\rho) =  (1-\lambda)\rho + \dfrac{\lambda}{d}\mathds{1},
\end{equation}
where $0 \leq \lambda \leq 1$, which admits a representation with Kraus operators $K_0 = \sqrt{1-\dfrac{3\lambda}{4}}\mathds{1}$, $K_1 = \sqrt{\dfrac{\lambda}{4}}X$, $K_2 = \sqrt{\dfrac{\lambda}{4}}Y$, and $K_3 = \sqrt{\dfrac{\lambda}{4}}Z$.
\par Let us first analyze the dephasing channel. Its model is shown in Fig. \ref{fig:dephasing}. We start with the following initial state
\begin{equation}
    \rho^{[0]} = p\ket{\psi}\bra{\psi}\otimes \ket{0}\bra{0} + (1-p) \ket{\psi}\bra{\psi}\otimes \ket{1}\bra{1}.
\end{equation} 
After one step, we have
\begin{equation}
\begin{split}
    \rho^{[1]} &= 1/2(p\ket{\psi}\bra{\psi} + (1-p)Z\ket{\psi}\bra{\psi} Z) \otimes \ket{0}\bra{0} \\
    & + 1/2((1-p)\ket{\psi}\bra{\psi} + p Z\ket{\psi}\bra{\psi} Z) \otimes \ket{1}\bra{1}.
\end{split}
\end{equation}
Post-selecting the result on the last node, we acquire the state
\begin{equation}
\rho' = (1-p) \ket{\psi}\bra{\psi} + p Z\ket{\psi}\bra{\psi}Z
\end{equation}
which is exactly the action of the dephasing channel. Note that the environment of the OQW induces the dephasing in this case and that the steady state is reached after a single step.
\begin{figure}[b]
\[\begin{tikzcd}
	{\ket{0}} && {\ket{1}}
	\arrow["{\sqrt{\lambda} \mathds{1}}", shift left=2, from=1-1, to=1-1, loop, in=55, out=125, distance=10mm]
	\arrow["{\sqrt{\omega}Z}", curve={height=-12pt}, from=1-1, to=1-3]
	\arrow["{\sqrt{\lambda} Z}", curve={height=-12pt}, from=1-3, to=1-1]
	\arrow["{\sqrt{\omega} \mathds{1}}", from=1-3, to=1-3, loop, in=55, out=125, distance=10mm]
\end{tikzcd}\]
\caption{\label{fig:dephasing} OQW model to the dephasing channel. $X,Y,Z$ are the Pauli matrices and $\omega + \lambda = 1$.
}
\end{figure}
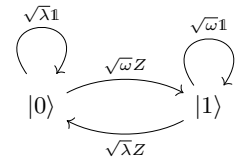

\par The depolarizing channel is modeled by the OQW in Fig. \ref{fig:depolarizing} with initial state 
\begin{equation}
\begin{split}
    \rho^{[0]} & = \left[\dfrac{\lambda}{4}\ket{\psi}\bra{\psi} + \dfrac{\lambda}{4} X\ket{\psi}\bra{\psi}X\right]\otimes \ket{0}\bra{0} \\
    & + \left[\dfrac{\lambda}{4}\ket{\psi}\bra{\psi}\right] \otimes \ket{1}\bra{1} \\
    & + \left[\left(1- \dfrac{3\lambda}{4}\right) \ket{\psi}\bra{\psi}\right]\otimes \ket{2}\bra{2}.
\end{split}
\end{equation}
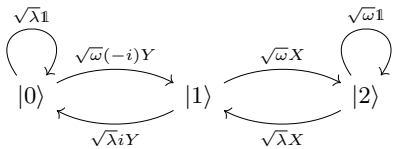
\begin{figure}[b]
\[\begin{tikzcd}
	{\ket{0}} && {\ket{1}} && {\ket{2}}
	\arrow["{\sqrt{\lambda} \mathds{1}}", from=1-1, to=1-1, loop, in=55, out=125, distance=10mm]
	\arrow["{\sqrt{\omega} (-i)Y}", curve={height=-12pt}, from=1-1, to=1-3]
	\arrow["{\sqrt{\lambda} i Y}", curve={height=-12pt}, from=1-3, to=1-1]
	\arrow["{\sqrt{\omega} X}", curve={height=-12pt}, from=1-3, to=1-5]
	\arrow["{\sqrt{\lambda} X}", curve={height=-12pt}, from=1-5, to=1-3]
	\arrow["{\sqrt{\omega} \mathds{1}}", from=1-5, to=1-5, loop, in=55, out=125, distance=10mm]
\end{tikzcd}\]
\caption{\label{fig:depolarizing} OQW model to the depolarizing channel. $X,Y,Z$ are the Pauli matrices and $\omega + \lambda = 1$.
}
\end{figure}
In the long-term run, the post-selected state in the last node is
\begin{equation}
\begin{split}
    \rho ' & = \dfrac{\lambda}{4}X\ket{\psi}\bra{\psi}X + \dfrac{\lambda}{4}Y\ket{\psi}\bra{\psi}Y  \\
    & +\dfrac{\lambda}{4}Z\ket{\psi}\bra{\psi}Z  + \left(1-\dfrac{3\lambda}{4} \right) \ket{\psi}\bra{\psi} 
    \end{split}
\end{equation}
which is exactly the action of the depolarizing channel. Observe that by exploiting the commutation relations of the Pauli matrices we have a more efficient implementation in terms of graph size compared to the general method described in \ref{subsec:convexcomb}. Mathematically, one can start from a more complex initial state and employ the same OQW used for the dephasing channel to construct the depolarizing channel; however, this approach is inefficient, as it incurs significant overhead. Note that both the depolarizing and dephasing channel results do not depend on $\omega$. This is because we are post-selecting the result measured in the last node. Consequently, we are erasing the information inherited from the walk process.

\section{Simulation Models}
\label{sec:simulation}
We now present a simulation protocol specifically designed for the OQW model of quantum computation. This protocol takes advantage of the inherent locality of the graph to allow for an efficient simulation process. We compare our circuit implementation with other methods, such as the Stinespring dilation and Sz.-Nagy dilation.

\subsection{Stinespring dilation}
\label{subsec:Stinespring}
\par One common method to simulate the dynamics of an open quantum system is to explore the Stinespring dilation \cite{Suri23}, which considers the evolution of the open quantum system as part of the unitary evolution of a larger system. Given Kraus operators $\{K_i\}_{i =1}^m$ for the evolution of a quantum state $\ket{\psi} \in \mathcal{H}'$, the Stinespring procedure gives us a unitary
\begin{equation}
    U_{St} = \begin{bmatrix}
        K_1 & \cdots & \cdots \\
        K_2 & \cdots & \cdots \\
        \vdots & \dots & \vdots \\
       K_m & \ddots & \dots \\
    \end{bmatrix}
\end{equation}
acting on a larger space $\mathcal{C} \otimes \mathcal{H}'$ (where $\mathcal{C}$ is a Hilbert space with dimension $m$, called ancillary system) satisfying
\begin{equation}
    U_{St}\ket{0} \otimes \ket{\psi} = \sum_{i=1}^{m} \ket{i} \otimes K_i \ket{\psi}.
\end{equation}
Note that when we trace out the ancillary system, we have exactly the evolution of the open quantum system that we wanted. Therefore, in order to implement $U_{St}$ in the circuit model, we need the system's dimension to be of the order $\mathcal{O}(dm)$, where $d = \text{dim}\mathcal{H}'$. In our linear OQW model, we have $m = 2|G|$ Kraus operators (two for each node), and they all act on $\mathcal{H}\otimes \mathcal{G}$, which has dimension $d = \text{dim}(\mathcal{H})|G|$.
\subsection{Sz.-Nagy dilation}
Suppose that, under the same conditions as in the Stinespring procedure, we wish to simulate the system's evolution by simulating the Kraus operators separately instead of all with the same unitary. This can be done with the Sz.-Nagy dilation procedure \cite{Suri23}: Given a $K_i$, we define a unitary
\begin{equation}
    U_{i,Sz} = \begin{bmatrix}
       K_i & \sqrt{I -K_i K_i^\dagger}  \\
        \sqrt{I - K_i^{\dagger}K_i} & -K_i^{\dagger}
    \end{bmatrix}.
\end{equation}
This single dilation has dimension $\mathcal{O}(d)$ (the dimension of the matrix is $2d$). Since we need one dilation for each of the $m$ Kraus operators, the total dimension needed is $\mathcal{O}(md)$, which is the same as needed in the Stinespring procedure.

\subsection{Dilation through locality}
\label{subsec:localsim}
We use the locality of the walk to propose a more efficient way to perform the simulation. We define the unitary 
$$U_{Loc}\colon \mathcal{H} \otimes \mathcal{G} \otimes \mathcal{A} \to \mathcal{H} \otimes \mathcal{G} \otimes \mathcal{A} $$
where $\mathcal{H}$ is the walker space, $\mathcal{G}$ is the graph space, $\mathcal{A}$ is an ancilla qubit, and
\begin{equation}
U_{Loc}\ket{\psi}\ket{i} \ket{j} = \begin{cases}
U_i\ket{\psi} \ket{i+1} \ket{1} &\text{if $i < N-1, j=1$};\\
U_{i-1}^\dagger\ket{\psi}\ket{i-1}\ket{0} &\text{if $i  >0, j=0$};\\
\ket{\psi}\ket{N-1}\ket{0} &\text{if $i = N-1, j = 1$}; \\
\ket{\psi}\ket{0}\ket{1} &\text{if i= 0, j = 0} .
\end{cases}
\end{equation}
This unitary is constructed with the following reasoning: If the ancilla is in state $\ket{1}$, the walker goes right or stays still in the right boundary (applying the respective unitary), and if the ancilla is in state $\ket{0}$, the walker goes left or stay still in the left boundary (applying the respective inverse unitary). However, in order to make $U_{Loc}$ unitary, we need to adjust the values of the ancilla qubit at the boundaries $i = 0, N-1$.

\par For each step, we need the qubit $\mathcal{A}$ initialized in the state $\rho_a = \lambda \ket{0}\bra{0} + \omega \ket{1}\bra{1}$. We then apply $U_{Loc}$ and trace $\mathcal{A}$. The state after repeating this procedure $n$ times is the state of the OQW after $n$ steps, as shown in Theorem \ref{thm:correctness} in Appendix \ref{app:proofs}. If we had an infinite (or circular) OQW, this same dilation structure would also work just by ignoring the boundary corrections. It is important to note that the state $\lambda \ket{0}\bra{0} + \omega \ket{1}\bra{1}$ can be initialized by performing a non-selective measurement on a qubit prepared in the state $\sqrt{\lambda}\ket{0} + \sqrt{\omega}\ket{1}$. \\
\par This dilation procedure can be easily generalized for an OQW following the following conditions:
\begin{enumerate}
    \item For each node $\ket{i}$, there are exactly a fixed number $1 \leq k \leq N$ of operators $B_i^j$;
    \item For all $i,j$, we have $B_i^j = \sqrt{w_{i,j}}U_{i,j}$, where $U_{i,j}$ is a unitary and $w_{i,j} \in [0,1]$ satisfying $\sum_j w_{i,j} = 1$;
    \item For $i,i'$, we have the set equality $\{w_{i,j}\}_j = \{w_{i',j}\}_j$.
\end{enumerate}
\par The only change to our algorithm is that we add more cases to the jump possibilities (we describe in detail a unitary with two possibilities: left or right) and adjust the jump probabilities using a qudit of dimension $k$ prepared in the state $\sqrt{w_0}\ket{0} + \cdots \sqrt{w_{k-1}} \ket{k-1}$ and performing a non-selective measurement, where $w_i$ are the proportion weight of the unitaries $U_{i,j}$. We obtain the Stinespring realization procedure in the worst-case scenario ($k = N$).

\subsection{Computational Complexity Analysis}
Note that we obtain a system with a smaller dimension than the other general simulation methods because we explore the locality of the graph data structures. In the Stinespring and Sz.-Nagy procedure, we need dimension $\mathcal{O}(n\text{dim}\mathcal{H}|G|^2)$, where $n$ is the number of steps in the walk, while the dimension of our proposed protocol is $\mathcal{O}(n\text{dim}\mathcal{H}|G|)$.
\par We can use well-established results on isometries \cite{Suri23, QCI,iten16} to estimate the circuit's depth and count of CNOT gates. In the worst-case scenario, the CNOT count and circuit depth are proportional. Therefore, they share an asymptotic upper bound. Now we analyze the asymptotic behavior of the CNOT count. Using column-by-column isometry \cite{iten16}, the CNOT count of the Stinespring dilation is $\mathcal{O}(md^2)$ which is $\mathcal{O}(\text{dim}\mathcal{H}^2 |G|^3)$ for one step in our case. The circuit CNOT count for one Kraus operator, using Sz.-Nagy dilation, is $\mathcal{O}(d^2)$, which is then $\mathcal{O}(md^2)$ to implement all $m$ Kraus operators. In our case, this reduces to $\mathcal{O}(\text{dim}\mathcal{H}^2 |G|^3)$, which is the exact CNOT count as in the Stinespring dilation. Note, however, that since all Kraus operators can be simulated in parallel using the Sz.-Nagy method, the asymptotic circuit depth is not as big as the Stinespring method \cite{Azevedo2025}. Using the same estimations, we see that our unitary $U_{Loc}$ has CNOT gate count $\mathcal{O}(\text{dim}\mathcal{H}^2 |G|^2)$, which is a substantial gain for large graphs or a large number of steps. For one step, in Stinespring and Sz.-Nagy, we needed dimension $\mathcal{O}(\text{dim}\mathcal{H}|G|^2)$, while for $U_{Loc}$ we only need $\mathcal{O}(\text{dim}\mathcal{H}|G|)$. Note that it is a linear gain of order $|G|$ for all quantities that we are measuring.

\subsection{Circuit Model}
Now we will show how to construct the circuits that can implement the OQW, in order to efficiently simulate this model in a perfect circuit-based quantum computer. The need for efficient ways to simulate this model comes from computational and physical points of view. For computation, it is interesting to know how to translate between different computational languages, while for physics, OQW represent an important class of physical systems which are types of quantum walk immersed in thermal bathes in which the bath induces the walk process \cite{Petr15, Petr14}.
\label{subsec:circuit}
We define the following unitaries
\begin{equation}
S \ket{i} = \begin{cases}
\ket{i-1} &\text{if $i < N-1$}\\
\ket{0} &\text{if $i = N-1$}
\end{cases}
\end{equation}
\begin{equation}
P \ket{i} = \begin{cases}
\ket{i-1} &\text{if $i > 0$}\\
\ket{N-1} &\text{if $i = 0$}
\end{cases}
\end{equation}
that are the successor and predecessor operators for our OQW. For the sake of simplicity, we assume that $|G| = 2^g$ and $\text{dim}\mathcal{H} = 2^h$. The general description of the algorithm can be found in Appendix \ref{app:general}. Each node $\ket{i}$ of the graph is encoded as its binary representation $\ket{i_{g-1}\dots i_0}$, where $i = \sum_k i_k 2^k$. In the circuit representation $S$ and $A$ are, respectively, as follows (for $g = 4$): 

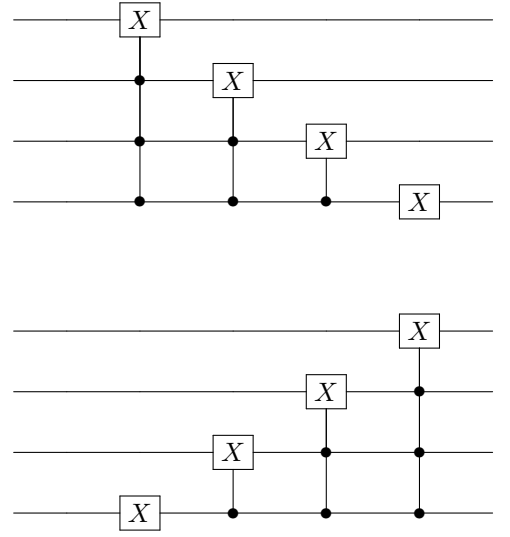
\[ \Qcircuit @C=2em @R=1em {
&  \qw & \gate{X} & \qw & \qw  & \qw & \qw\\
&  \qw & \ctrl{-1} & \gate{X} & \qw & \qw & \qw\\
&  \qw & \ctrl{-2} & \ctrl{-1} & \gate{X} & \qw & \qw\\
& \qw & \ctrl{-3} & \ctrl{-2} & \ctrl{-1} & \gate{X} & \qw 
} \]

\[ \Qcircuit @C=2em @R=1em {
& \qw & \qw & \qw & \qw & \gate{X} & \qw \\
&  \qw & \qw & \qw & \gate{X} & \ctrl{-1} & \qw\\
& \qw & \qw & \gate{X} & \ctrl{-1} & \ctrl{-2} & \qw \\
&  \qw & \gate{X} & \ctrl{-1} & \ctrl{-2} & \ctrl{-3} & \qw
} \]

\par The main idea behind $U_{Loc}$ is to conditionally apply unitaries based on the position of the walker, so we need to transport this idea to the circuit model. For example, suppose $g = 2, h = 1$. First, we define a circuit $\mathcal{R}$ that represents the right movements corresponding to the nodes on the graph ($q_H$ is the register for $\mathcal{H}$, $q_{G,1},q_{G,2}$ are the registers for $\mathcal{G}$, $q_A$ is the register for $\mathcal{A}$ and $q_{A'}$ is the register for $\mathcal{A}'$):

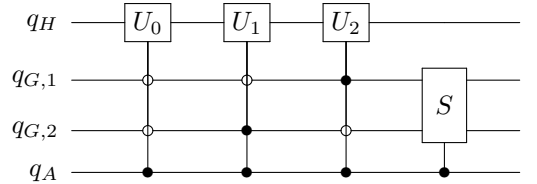
\[ \Qcircuit @C=2em @R=1em {
& \lstick{q_H} &  \gate{U_0} & \gate{U_1} & \gate{U_2}  &  \qw  & \qw        \\
& \lstick{q_{G,1}}  & \ctrlo{-1} & \ctrlo{-1} & \ctrl{-1}  &    \multigate{1}{S} & \qw          \\
& \lstick{q_{G,2}} &   \ctrlo{-2} & \ctrl{-2} & \ctrlo{-2}  &   \ghost{S} & \qw  \\
& \lstick{q_A} &  \ctrl{-3} & \ctrl{-3} & \ctrl{-3} & \ctrl{-1} & \qw 
}  \]
\\
\\
In the same way, we define a circuit $\mathcal{L}$ for the left movements:

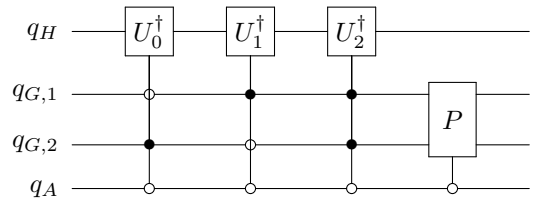
\[ \Qcircuit @C=2em @R=1em {
& \lstick{q_H} &  \gate{U_0^\dagger} & \gate{U_1^\dagger} & \gate{U_2^{\dagger}}  &  \qw  & \qw        \\
& \lstick{q_{G,1}} &  \ctrlo{-1} & \ctrl{-1} & \ctrl{-1}  &    \multigate{1}{P} & \qw          \\
& \lstick{q_{G,2}} & \ctrl{-2} & \ctrlo{-2} & \ctrl{-2}  &   \ghost{P} & \qw  \\
& \lstick{q_A} &  \ctrlo{-3} & \ctrlo{-3} & \ctrlo{-3} & \ctrlo{-1} & \qw 
} \]

\par Now, we need to adjust the circuit output in the boundaries. We will use an additional qubit $\mathcal{A}'$ initialized in the state $\ket{0}$ to implement the boundary terms in the circuit. This extra qubit does not increase the previous asymptotic analysis, since it is an extra qubit per step. Therefore, this circuit is not the direct implementation of $U_{Loc}$, but a slight modification. For the movements to the right, we have the circuit $\mathcal{R}_{\text{b}}$:
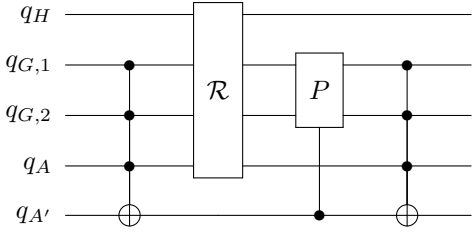
\[ \Qcircuit @C=2em @R=1em {
& \lstick{q_H} & \qw & \multigate{3}{\mathcal{R}} & \qw & \qw &  \qw      \\
& \lstick{q_{G,1}} & \ctrl{3} & \ghost{\mathcal{R}} & \multigate{1}{P}  & \ctrl{3} & \qw \\
& \lstick{q_{G,2}} & \ctrl{2} & \ghost{\mathcal{R}} &  \ghost{P} & \ctrl{2}   & \qw \\
& \lstick{q_A} & \ctrl{1} & \ghost{\mathcal{R}} & \qw & \ctrl{1}  & \qw \\
& \lstick{q_{A'}} & \targ & \qw & \ctrl{-2}  & \targ & \qw
}\]
where the subscript $b$ stands for "boundary", indicating that it is the circuit with terms correcting the boundary movements. For the movements to the left, we have the circuit $\mathcal{L}_{\text{b}}$:
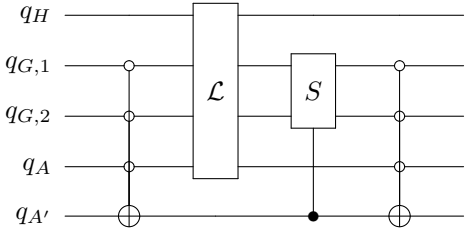
\[ \Qcircuit @C=2em @R=1em {
& \lstick{q_H} & \qw & \multigate{3}{\mathcal{L}} & \qw & \qw &  \qw      \\
& \lstick{q_{G,1}} & \ctrlo{3} & \ghost{\mathcal{L}} & \multigate{1}{S}  & \ctrlo{3} & \qw \\
& \lstick{q_{G,2}} & \ctrlo{2} & \ghost{\mathcal{L}} &  \ghost{S} & \ctrlo{2}   & \qw \\
& \lstick{q_A} & \ctrlo{1} & \ghost{\mathcal{L}} & \qw & \ctrlo{1}  & \qw \\
& \lstick{q_{A'}} & \targ & \qw & \ctrl{-2}  & \targ & \qw
}\]

\par Our circuit for one step is then the concatenation of $\mathcal{R}_b$ with $\mathcal{L}_b$ after the initialization of $q_A$ in state $\sqrt{\lambda}\ket{0} + \sqrt{\omega}\ket{1}$ and a non-selective measurement in this same qubit. After this step, we trace the qubits $\mathcal{A}$ and $\mathcal{A}'$ and use another two to repeat the procedure. 
\par One can find the detailed description of the algorithm for the most general case in Appendix \ref{app:general}. 

\par We now give a schematic representation to summarize the one-step procedure:
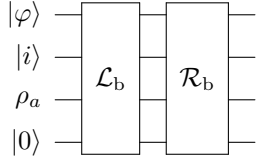
\[ \Qcircuit @C=1em @R=.7em {
& \lstick{\ket{\varphi}}  & \multigate{3}{\mathcal{L}_\text{b}} & \multigate{3}{\mathcal{R}_\text{b}} & \qw \\
& \lstick{\ket{i}} & \ghost{\mathcal{L}_\text{b}} & \ghost{\mathcal{R}_\text{b}} & \qw \\
& \lstick{\rho_a}  & \ghost{\mathcal{L}_\text{b}} & \ghost{\mathcal{R}_\text{b}} & \qw \\
& \lstick{\ket{0}} & \ghost{\mathcal{L}_\text{b}} & \ghost{\mathcal{R}_\text{b}} & \qw
}\]
where $\rho_a = \lambda \ket{0}\bra{0} + \omega \ket{1}\bra{1}$, $\ket{\varphi} \in \mathcal{H}$ is the state of the walker and $\ket{i} \in \mathcal{G}$ is the node in which the walker is located. This circuit works as follows:
\begin{enumerate}
    \item The part $\ket{0}\bra{0}$ of the mixture $\rho_a$ induces only $\mathcal{L}_{\text{b}}$, making the state $\ket{\varphi}\ket{i}$ change to $U_{i-1}^\dagger\ket{\varphi}\ket{i-1}$ if $i > 1$ and to $\ket{\varphi}\ket{0}$ if $i = 0$;
    \item The part $\ket{1}\bra{1}$ of the mixture $\rho_a$ induces only $\mathcal{R}_{\text{b}}$, making the state $\ket{\varphi}\ket{i}$ change to $U_{i}\ket{\varphi}\ket{i+1}$ if $i < N-1$ and to $\ket{\varphi}\ket{N-1}$ if $i = N-1$.
\end{enumerate}
Therefore, starting the register $\mathcal{A}$ as $\rho_a$ makes the circuit change from $\ket{\varphi}\ket{i}$ to its corresponding state after evolving one step in the open quantum walk.

\subsection{Simulation results}
Here, we simulate the evolution of the OQW model of quantum computation choosing random unitaries in Qiskit with our proposed circuit implementation and compare the system's dimension, CNOT and one-qubit gate count and circuit's depth with the evolution simulated using the Stinespring method, as shown in Tables \ref{tab:table1}, \ref{tab:table2}, \ref{tab:table3}. No optimization was used in the \texttt{transpile} routine \cite{Qiskit2024,QiskitTerraTranspile}, and the transpiled circuit only has CNOTs and one-qubit gates. The Stinespring method was simulated using the method  \texttt{Isometry} from Qiskit~\cite{Qiskit2024,QiskitTerraIsometry}. Other state of the art simulation methods based on Singular Value Decomposition \cite{Suri23} and Sz.-Nagy dilation are very different from our method; they need either a classical postprocessing or are probabilistic methods, while Stinespring method and our local unitary method are expressed as a large unitary that produces the desired quantum state deterministically after tracing out the ancillary system, without the need of any classical postprocessing to construct the density matrix. That is why we chose to compare only with the closest method. It is important to notice that there are efforts in combining efficiently the Sz.-Nagy dilation into one circuit \cite{Azevedo2025}. For $|G| = 2$, we have the CNOT and one-qubit gate count of order $10^4$ and the circuit depth of order $2\times 10^4$ for both methods. The system's dimension needed for Stinespring's method is exactly twice the dimension of our proposed circuit. This means that the circuit for Stinespring will have one qubit more than the circuit of our method. For $|G| = 4$ we still do not see gains other than a reduction in the dimension of the system, with both the CNOT count and the one-qubit gate count of order $8\times 10^4$, and both depths are around $1.5\times 10^5$. The dimension for Stinespring method in this case is exactly four times the dimension of our proposed circuit, meaning that Stinespring's circuit has 2 more qubits than ours. For $|G| = 16$, the gains are visible: The Stinespring circuit has CNOT and one-qubit gate count of order $1.3\times 10^6$, while our method has both counts of order $3 \times 10^5$. The depth of the Stinespring circuit is of the order $2.7\times 10^6$, while our method is $6\times 10^5$. The dimension of Stinespring's circuit is eight times the dimension of our proposed circuit, meaning that it has 3 qubits more than ours. 

\par Since classical simulation does not allow us to simulate large graphs with a large number of steps, the circuit's depth and the CNOT count appear with small gains for $|G| = 2,4$. Nevertheless, the reduction of dimension is visible even in small graphs. We maintain $\text{dim}\mathcal{H} = 2$ fixed throughout our simulations because the gains of our model in the system's dimension, CNOT count, and circuit depth are given in terms of $|G|$, i.e., the contribution of $\text{dim}\mathcal{H}$ for the metrics that we analyzed (CNOT count, circuit depth, and one-qubit gate count) is the same for Stinespring's method and our method.

\begin{table*}
\caption{\label{tab:table1} Comparison between Stinespring method of simulation and our proposed method for $\text{dim}(\mathcal{H}) = 2, |G| = 4, \omega = 0.6$. The unitaries $U_i$ were chosen randomly. By our estimative, we use $n_{steps} = 21$.}
\begin{ruledtabular}
\begin{tabular}{ccddd}
 &CNOT Count &\mbox{One-qubit gate count}&\mbox{Circuit's Depth}&\mbox{System's Dimension}\\
\hline
Stinespring &10458&11676&20769&1344\\
$U_{Loc}$ &10206& 12831 & 19488 & 672\\
\end{tabular}
\end{ruledtabular}
\end{table*}

\begin{table*}
\caption{\label{tab:table2} Comparison between Stinespring method of simulation and our proposed method for $\text{dim}(\mathcal{H}) = 2, |G| = 8, \omega = 0.7$. The unitaries $U_i$ were chosen randomly. By our estimative, we use $n_{steps} = 21$.}
\begin{ruledtabular}
\begin{tabular}{ccddd}
 &CNOT Count&\mbox{One-qubit gate count}&\mbox{Circuit's Depth}&\mbox{System's Dimension}\\
\hline
Stinespring &87549& 91308& 175203 & 5376\\
$U_{Loc}$ &74592& 87381 & 144446 & 1344 \\
\end{tabular}
\end{ruledtabular}
\end{table*}

\begin{table*}
\caption{\label{tab:table3} Comparison between Stinespring method of simulation and our proposed method for $\text{dim}(\mathcal{H}) = 2, |G| = 16, \omega = 0.7$. The unitaries $U_i$ were chosen randomly. By our estimative, we use $n_{steps} = 41$.}
\begin{ruledtabular}
\begin{tabular}{ccddd}
 &CNOT Count&\mbox{One-qubit gate count}&\mbox{Circuit's Depth}&\mbox{System's Dimension}\\
\hline
Stinespring &1363537& 1382930& 2728550 & 41984\\
$U_{Loc}$ &293888& 334601 & 570728 & 5248 \\
\end{tabular}
\end{ruledtabular}
\end{table*}

\section{Discussion and Conclusion}
\label{sec:conclusion}
We have further developed the theory of time evolution of the OQW model of quantum computation, giving us an estimate on the number of steps $n_{steps}$ to be reasonably close to the steady-state, which is fundamental to design experiments and computationally feasible quantum simulations. Moreover, we have developed a way to canonically embed random unitary evolutions in the OQW framework, including important quantum channels used for modeling noise and decoherence, such as dephasing and depolarizing channels.
\par We also provided an efficient way to simulate the dynamics of this model (and showed how this can be generalized to a larger class of OQW) based on the locality property of the underlying graph, reducing depth, CNOT count, one-qubit gate count, and dimension by a factor of $|G|$, which is a considerable gain. Furthermore, we developed the circuit description and showed a substantial gain in the simulation results, opening paths to make the implementation of this type of dynamics driven on graphs feasible on real quantum computers.
\par Although the simulation results strongly support our theoretical analysis, it is important to state that we did not simulated the generalization of this protocol which was described at the end of Section \ref{subsec:localsim}. It is reasonable to expect that if the graph is fully connected, then locality loses its meaning, meaning that the generalized protocol should have worse results. The formalization and simulation of these claims is work to be done.
\par We also believe that the circuit can be much more optimized. This is due to the fact that we used the standard Qiskit transpilation methods to implement the multi-controlled unitaries for each step, while there are better optimization methods for this type of situation \cite{rosa25}. It is also our interest to dig deeper into some fundamental questions that are still not completely clear, such as the connection between OQW and quantum walks with decoherence and a more refined analysis of the evolution in more general topologies.

\begin{acknowledgments}
This study was financed in part by the Coordenação de Aperfeiçoamento de Pessoal de Nível Superior – Brazil (CAPES) – Finance Code 001.
\par N.K.B. acknowledges financial support from CNPq Brazil (442429/2023-1) and FAPESP(Grant 2021/06035-0).
\end{acknowledgments}
\appendix

\section{Mathematical statements and proofs}
\label{app:proofs}
Here we explore the formal statements and proofs of some claims that have been made in the text.

\begin{theorem}
    For $p_{suc}\geq \eta$, it suffices that $\omega \geq \dfrac{1}{2-\eta}$.
\end{theorem}
\begin{proof}
    \begin{equation*}
        \begin{split}
            p_{suc} & = \dfrac{a^{N-1}(a-1)}{a^N-1} \geq \eta \\
            & \implies a^N-a^{N-1}\geq \eta a^N-\eta \\
            & \implies [(1-\eta)a^N - a^{N-1}]\geq -\eta \\
            & \implies a^{N-1}[(1-\eta)a - 1] \geq -\eta
        \end{split}
    \end{equation*}
Since $-\eta \leq 0 $ and $a^{N-1}\geq 0$, for the inequality to hold, it is sufficient to ask for 
\begin{equation}
(1-\eta)a-1 \geq 0.
\end{equation}
Since $a = \dfrac{\omega}{1-\omega}$, we obtain that it is sufficient to ask for
\begin{equation}
    \dfrac{\omega}{1-\omega} \geq \dfrac{1}{1-\eta}
\end{equation}
which is equivalent to
\begin{equation}
    \omega \geq \dfrac{1}{2-\eta}.
\end{equation}
\end{proof}

\begin{theorem}
\label{thm:correctness}
\textbf{(Correctness of the algorithm)} The realization procedure with $U_{Loc}$ in Section $\ref{subsec:localsim}$ gives us the correct dynamics for the OQW model of quantum computation.
\end{theorem}
\begin{proof}
Let us call the $n$-th iteration of our realization procedure $L^{[n]}$. We wish to prove that $L^{[n]} = \rho^{[n]}, \forall n$. We proceed inductively. The cases $n=0,1$ are immediate. We know that $\rho^{[n]}$ has the form of Equation $\ref{eqn:niteration}$. Assume that $L^{[n]} = \rho^{[n]}$. For the next step, we have
\begin{equation}
    \rho^{[n+1]} = \sum_{i = 0}^{N-1} p_i' U_{i-1}\cdots U_0\ket{\psi}\bra{\psi}U_0^\dagger \cdots U_{i-1}^\dagger \otimes \ket{i}\bra{i}
\end{equation}
where $p_0' = \lambda p_0 + \lambda p_1$, $p_{N-1}' = \omega p_{N-2} + \omega p_{N-1}$, and for $0 < i < N-1$ we have
\begin{equation}
    p_i' = \omega p_{i-1} + \lambda p_{i+1}.
\end{equation}
Performing the local realization procedure, we have the following (where $\rho_a = \lambda \ket{0}\bra{0} + \omega \ket{1}\bra{1}$ is the ancillary qubit):
\begin{equation}
\begin{split}
    L^{[n+1]} & = \text{Tr}_2[ U_{Loc} (L^{[n]} \otimes \rho_a)U_{Loc}^\dagger]\\
    & = \text{Tr}_2[U_{Loc} (\rho^{[n]} \otimes \rho_a)U_{Loc}^\dagger].
\end{split}
\end{equation}
We obtain Equations \ref{eqn:big1}  and \ref{eqn:big2} from the direct application of $U_{Loc}$ to our state.
\begin{widetext}
\begin{align}
\label{eqn:big1}
    U_{Loc} \rho^{[n]}\otimes \lambda \ket{0}\bra{0}U_{Loc}^\dagger   = \lambda p_0 \ket{\psi}\bra{\psi}\otimes \ket{0}\bra{0} \otimes \ket{1}\bra{1} + \sum_{i > 0} \lambda p_i \ket{\psi}^{[n-2]}\bra{\psi}^{[n-2]} \otimes \ket{i-1}\bra{i-1} \otimes \ket{0}\bra{0}.
\end{align}

\begin{align}
\label{eqn:big2}
    U_{Loc} \rho^{[n]}\otimes \omega \ket{1}\bra{1}U_{Loc}^\dagger   = \omega p_{N-1} \ket{\psi}\bra{\psi}\otimes \ket{N-1}\bra{N-1} \otimes \ket{0}\bra{0}  + \sum_{i< N-1} \omega p_i \ket{\psi}^{[n]}\bra{\psi}^{[n]} \otimes \ket{i+1}\bra{i+1} \otimes \ket{1}\bra{1}.
\end{align}
\end{widetext}

By a direct comparison we find that the new probability weight for $L^{[n+1]}$ are exactly the corresponding weight of $\rho^{[n+1]}$. The result then follows by induction.
\end{proof}

\begin{theorem}
    The OQW model of quantum computation with $N = 2$ stabilizes after just one step.
\end{theorem}
\begin{proof}
    We already saw that 
\begin{equation}
\begin{split}
\rho^{[1]} & = [\lambda p \rho_0 + (1-p)\lambda U^{\dagger} \rho_1 U]\otimes \ket{0}\bra{0} \\
& + [\omega p U\rho_0 U^\dagger + \omega (1-p) \rho_1] \otimes \ket{1}\bra{1}.
\end{split}
\end{equation}
Applying the evolution operator again, we have
\begin{equation}
\begin{split}
\rho^{[2]} & = [\lambda^2 p \rho_0 + (1-p)\lambda^2 U^{\dagger} \rho_1 U]\otimes \ket{0}\bra{0} \\
& + [\lambda\omega p \rho_0  + \lambda\omega (1-p) U^{\dagger}\rho_1 U] \otimes \ket{0}\bra{0} \\
& + [\omega^2 p U \rho_0 U^\dagger + \omega^2(1-p)\rho_1] \otimes \ket{1}\bra{1} \\
& + [\omega \lambda p U \rho_0 U^\dagger + \omega \lambda(1-p)\rho_1] \otimes \ket{1}\bra{1}.
\end{split}
\end{equation}
Summing each correspondent term and using $\lambda^2 + \lambda \omega = \lambda (\lambda +\omega) = \lambda$ and $\omega^2 + \omega\lambda = \omega(\omega + \lambda) = \omega$, we have $\rho^{[2]} = \rho^{[1]}$.
\end{proof}

\section{Description of algorithm} 
\label{app:general}
In this section, we develop the algorithms for each component of the circuit model, just as described in Section \ref{subsec:circuit}. We implemented it on Qiskit, but it can be implemented in any Quantum Circuit environment. All quantum registers start in the state $\ket{0}$. Here we define 
\begin{equation}
    U_{\omega} = \begin{bmatrix}
       \sqrt{\lambda} & \sqrt{\omega}  \\
        \sqrt{\omega} & -\sqrt{\lambda}
    \end{bmatrix} = R_y(-2 \cdot \text{arcsin}(\sqrt\omega)),
\end{equation}
where $R_y$ is the spin rotation matrix around the $y$-axis 
\begin{equation}
    R_y(\theta) = e^{i \frac{\theta}{2}Y} = \begin{bmatrix}
      \text{cos}\frac{\theta}{2} &  \text{sin}\frac{\theta}{2}  \\
         -\text{sin}\frac{\theta}{2} & \text{cos}\frac{\theta}{2}
    \end{bmatrix},
\end{equation}
where $Y$ is the Pauli matrix
\begin{equation}
    Y = \begin{bmatrix}
      0 &  -i  \\
         i & 0
    \end{bmatrix}.
\end{equation}
\par We proceed to provide high-level pseudo-algorithm descriptions in modules for implementation of the circuits developed in Section \ref{subsec:circuit} in any quantum circuit framework, such as Qiskit, Circ, PennyLane, Strawberry Fields. We look forward to the implementation in Strawberry Fields and PennyLane due to potential applications of dissipative quantum computation in Machine Learning and Photonic Quantum Computation. One of our aims in this paper is to prove that our method has a better depth, CNOT count, and a lower dimension than the Stinespring method, so we implemented it only in Qiskit as a proof of concept.

\begin{algorithm}[H]
\caption{Initialize registers and unitaries}
\label{alg:register}
\begin{algorithmic}[1]
\State \textbf{initialize circuit} Start a quantum circuit with four quantum registers (i.e., four sets of qubits) $q_{H,i}$ ($1 \leq i \leq h$), $q_{G,j}$ ($1 \leq j \leq g$), $q_A$ and $q_{A'}$ for encoding the walker Hilbert space, the graph Hilbert space and the auxiliary qubits $A$ and $A'$ respectively. This registers have $h = \lceil \text{log}(\text{dim}(\mathcal{H})) \rceil$, $g = \lceil \text{log}(\text{dim}(\mathcal{G})) \rceil$, $1$ and $1$ qubits, respectively
\State \textbf{encoding spaces} Each basis vector $\ket{i}$ of each of the spaces is going to be encoded as the vector $\ket{i_{l-1}\dots i_0}$ in the respective quantum register $(l = h,g)$, where $i = \sum_k i_k 2^k$
\State \textbf{enconding unitaries} For each unitary $U_i$, we extend it to $U_i' = U_i \otimes \mathds{1}_{2^g-\text{dim}(\mathcal{H})}$ since $2^g$ can be bigger than $\text{dim}(\mathcal{H})$
\end{algorithmic}
\end{algorithm}
\begin{algorithm}[H]
\caption{Circuit $\mathcal{R}$}
\label{alg:R}
\begin{algorithmic}[1]
\State \textbf{initialize} Initialize registers and unitaries
\For{$\ket{j} = \ket{j_{g-1}\dots j_0}$ in register $q_{G,j}$ and $j < N-1$}
\State Add a multi-controlled $U_j'$ gate where the controls are the registers $q_{G,j_k}$ and $q_A$, and where the control bit in $q_A$ is the value $j_k$. The target qubits are the qubits in the register $q_{H,i}$
\EndFor

\State Add a multi-controlled $P$ gate where the control is $q_A$ with control bit 1, and the target are the qubits in $q_{G,j}$
\end{algorithmic}
\end{algorithm}
\begin{algorithm}[H]
\caption{Circuit $\mathcal{R}_b$}
\label{alg:Rb}
\begin{algorithmic}[1]
\State \textbf{initialize} Initialize registers and unitaries
\State Add a multi-CNOT gate, with control qubits the registers $q_{G,j}$ and $q_A$ (control bit is 1 for all qubits), and with target qubit the qubit $q_{A'}$
\State Add the circuit $\mathcal{R}$ in the registers $q_{H,i}$, $q_{G,j}$,$q_A$

\State Add a controlled $A$ gate with control $q_{A'}$ (control bit is 1) and target the register $q_{G,j}$
\State Add a multi-CNOT gate, with control qubits the registers $q_{G,j}$ and $q_A$ (control bit is 1 for all qubits), and with target qubit the qubit $q_{A'}$
\end{algorithmic}
\end{algorithm}
\begin{algorithm}[H]
\caption{Circuit $\mathcal{L}$}
\label{alg:L}
\begin{algorithmic}[1]
\State \textbf{initialize} Initialize registers and unitaries
\For{$\ket{j} = \ket{j_{g-1}\dots j_0}$ in register $q_{G,j}$ and $j > 0$}
\State Add a multi-controlled $U_j'^\dagger$ gate where the controls are the registers $q_{G,j_k}$  and $q_a$, and where the control bit in $q_a$ is $0$ and in $g_{G,j_k}$ is the value $j_k$. The target are the qubits in the register $q_{H,i}$
\EndFor

\State Add a multi-controlled $A$ gate where the control is $q_A$ with control bit 0, and the target are the qubits in $q_{G,j}$
\end{algorithmic}
\end{algorithm}
\begin{algorithm}[H]
\caption{Circuit $\mathcal{L}_b$}
\label{alg:Lb}
\begin{algorithmic}[1]
\State \textbf{initialize} Initialize registers and unitaries
\State Add a multi-CNOT gate, with control qubits the registers $q_{G,j}$ and $q_A$ (control bit is 0 for all qubits), and with target qubit the qubit $q_{A'}$
\State Add the circuit $\mathcal{L}$ in the registers $q_{H,i}$, $q_{G,j}$,$q_A$
\State Add a controlled $P$ gate with control $q_{A'}$ (control bit is 1) and target the register $q_{G,j}$

\State Add a multi-CNOT gate, with control qubits the registers $q_{G,j}$ and $q_A$ (control bit is 0 for all qubits), and with target qubit the qubit $q_{A'}$
\end{algorithmic}
\end{algorithm}
\begin{algorithm}[H]
\caption{Circuit for one step}
\label{alg:step}
\begin{algorithmic}[1]
\State \textbf{initialize} Initialize registers and unitaries
\State Add the unitary $U_\omega$ to qubit $q_A$
\State Add a non-selective measurement to qubit $q_A$
\State Append the circuit $\mathcal{R}_b$
\State Append the circuit $\mathcal{L}_b$
\end{algorithmic}
\end{algorithm}
\begin{algorithm}[H]
\caption{Simulation of $n$ steps}
\label{alg:locsim}
\begin{algorithmic}[1]
\State \textbf{initialize} Initialize registers and unitaries (here we need $n$ qubits in register $q_A$ and $n$ in register $q_{A'}$)
\For{$1 \leq i \leq N$}
\State Append the circuit for one step to the registers $q_{H,i}$, $q_{G,j}$, the $i$th qubit of $q_A$ and the $i$th qubit $q_{A'}$
\EndFor
\State Trace out the $2n$ qubits $q_A$, $q_{A'}$
\end{algorithmic}
\end{algorithm}

\nocite{*}

\bibliography{apssamp}

\end{document}